\documentclass[11pt]{article}

\usepackage{color}
\usepackage{epsfig}
\usepackage{amsmath}
\usepackage{amsfonts}
\usepackage{amssymb}
\usepackage{enumerate}
\usepackage{graphics}
\usepackage{color}
\usepackage{times}

\usepackage[ruled, vlined, linesnumbered, noresetcount]{algorithm2e}

\DeclareMathOperator*{\argmin}{arg\,min}

\newtheorem{example}{Example}

\newtheorem{theorem}{Theorem}

\newtheorem{remark}{Remark}

\newtheorem{lemma}{Lemma}
\newtheorem{definition}{Definition}

\newcommand{\remove}[1]{}

\def\T{{{\cal T}}}

\def\X{{\cal X}}
\def\m{{\mathfrak{m}}}
\def\k{{\mathfrak{K}}}
\def\VAL{{\mathit{VALUE}}}
\def\IND{{\mathit{INDEX}}}

\newcommand{\qed}{\hfill ~$\square$\bigskip}
\newcommand{\proof}{\noindent{\bf Proof.} }

\newcommand{\N}{{\mathbb{N}}}
\newcommand{\Infl}{{{\sf Influenced}}}

\newcommand{\squishlisttwo}{
 \begin{list}{$-$}
  { \setlength{\itemsep}{0pt}
    \setlength{\parsep}{0pt}
    \setlength{	\topsep}{0pt}
    \setlength{\partopsep}{0pt}
    \setlength{\leftmargin}{2em}
    \setlength{\labelwidth}{1.5em}
    \setlength{\labelsep}{0.5em} } }

\newcommand{\squishend}{
  \end{list}  }

\begin{document}

\date{}
\title{Time-Bounded Influence Diffusion with  Incentives}
\author{G. Cordasco\thanks{Department of  Psychology, University of Campania  ``L.Vanvitelli'', Italy},  L. Gargano\thanks{Department of Computer Science, University of Salerno, Italy}, J.G. Peters\thanks{School of Computing Science, Simon Fraser University, Canada}, A.A. Rescigno$^\dag$, U. Vaccaro$^\dag$}
\maketitle

\begin{abstract}
A widely studied model of influence diffusion in social networks represents the
network as a graph $G=(V,E)$ with an influence threshold $t(v)$ for each node.
Initially the members of an initial set $S\subseteq V$ are
influenced. During each subsequent round, the set of influenced nodes is
augmented by including every  node $v$ that has at least $t(v)$
previously influenced neighbours.
The general problem is to find a small initial set that influences the whole
network.
In this paper we extend this model by using \emph{incentives} to reduce the
thresholds of some nodes.
The goal is to minimize the total of the incentives required to ensure that
the process completes within a given number of rounds.
The problem is hard to approximate in general networks.
We present polynomial-time algorithms for paths, trees, and complete
networks.
\end{abstract}

\section{Introduction}

The \emph{spread of influence}
in social networks is the process by which individuals adjust their opinions,
revise their beliefs, or change their behaviours as a
result of  interactions with others (see  \cite{EK} and references therein quoted).
For example, \emph{viral marketing} takes advantage of peer influence among members of social networks for marketing \cite{DR-01}. The essential idea is that   companies  wanting to
promote products or behaviours might try to target and convince
a few individuals initially who will then trigger
a  cascade of further adoptions.
The intent of maximizing the spread of viral information across a network 
has suggested  several interesting  optimization  problems
 with various adoption paradigms. We refer to 
 \cite{CLC} for a recent 
discussion of the area.
In the rest of this section, we will explain and motivate our model of
information diffusion, 
describe our results, and discuss how
they relate to the existing literature.

\subsection{The Model}

A social network is a graph $G = (V,E)$, where the node set $V$ represents  the members of the network and $E$ represents the relationships among members.
We denote by $n=|V|$ the number of nodes, by  $N(v)$  the neighbourhood  of  $v$, and by 
 $d(v)=|N(v)|$ the degree of $v$, for each node  $v\in V$.

Let $t: V \to\N = \{1,2,\ldots\}$ be a function assigning integer thresholds
to the nodes of $G$;
we assume w.l.o.g.\ that $1\le t(v)\le d(v)$ holds for all $v\in V$.
For each node $v\in V$, the value $t(v)$ quantifies how hard it is to influence
 $v$, in the sense that easy-to-influence elements of the network have
``low'' $t(\cdot)$ values, and hard-to-influence elements have ``high''
$t(\cdot)$ values \cite{Gr}.
An {\em influence process in $G$} starting from a set $S\subseteq V$ of
initially influenced nodes is a sequence
of node subsets\footnote{We will omit the subscript $G$ whenever the graph
$G$ is clear from the context.},

\noindent $\Infl_G[S,0] = S$ 

\noindent $\Infl_G[S,\ell] = \Infl_G[S,\ell-1]\cup \Big\{v \,:\, \big|N(v)\cap \Infl_G[S,\ell - 1]\big|\ge t(v)\Big\},$ $\ell > 0$.
\\
Thus, in each round $\ell$, the set of influenced nodes is augmented by
including every uninfluenced node $v$ for which the number of \emph{already}
influenced neighbours is at least as big as $v$'s threshold $t(v)$.
We say that $v$ is influenced {\em at} round $\ell>0$ if
$v \in \Infl_G[S,\ell]\setminus \Infl_G[S,\ell - 1]$.
A target set for $G$ is a set $S$ such that it will influence the whole network,
that is, $\Infl_G[S,\ell]=V$, for some $\ell \geq 0$.

\smallskip
The classical {\em Target Set Selection   {\sc (TSS)}} problem having  as input a 
 network $G=(V,E)$ with thresholds $t:V\longrightarrow \mathbb{N}$, asks for 
 a  target set $S\subseteq V$ of \emph{minimum} size for $G$ \cite{ABW-10,Cic+}.
The TSS problem has roots in the general study
of the spread of influence in social networks (see~\cite{CLC,EK}).
For instance, in the area of viral marketing~\cite{DR-01}, companies  wanting to
promote products or behaviors might try to initially convince
a small  number of  individuals  (by offering  free samples or monetary rewards) who will then trigger
a cascade of influence in the social network leading to
the adoption  by a much larger number of individuals.

In this paper, we extend the classical model to make it more realistic. 
It was first observed in ~\cite{Dem14}  that the classical model limits
the optimizer to a binary choice between zero or
complete influence on each individual whereas customized incentives
could be 
more effective in realistic scenarios.
For example, a company promoting a new product may find that offering one hundred  free samples
is far less effective than offering a ten percent discount to one thousand   people. 

Furthermore, the papers mentioned above do not consider the time (number of rounds) necessary
to complete the influence diffusion process.
This could be quite  important in viral marketing;
a company may want to
influence its potential customers quickly before other companies can market
a competing product.

\smallskip
With this motivation, we formulate our model as follows.
%
An assignment of  incentives to the nodes
of a network $G=(V,E)$
is a function
$p: V \to\N_0 = \{0,1,2,\ldots\}$, where $p(v)$ is the amount of influence
initially applied on $v\in V$.
The effect of applying the incentive $p(v)$ on node $v$ is to decrease its
threshold, i.e., to make $v$ more susceptible to future influence.
It is clear that to start the process, there must be some nodes
for which the initially applied influences are at least as large as their
thresholds.  We assume, w.l.o.g., that $0 \leq p(v) \leq t(v) \leq d(v)$.
An influence process in $G$ starting with incentives given by a
function $p: V \to\N_0 = \{0,1,2,\ldots\}$ is a sequence
of node subsets 

\smallskip

\noindent $\Infl[p,0] = \{v \,:\, p(v)=t(v)\}$ 

\noindent $\Infl[p,\ell] = \Infl[p,\ell-1]\cup \Big\{v \,:\, \big|N(v)\cap \Infl[p,\ell - 1]\big|\ge t(v)-p(v)\Big\},$ $\ell > 0$.
\\
The cost of the incentive function $p:V\longrightarrow \N_0$ is $\sum_{v\in V} p(v)$.

Let $\lambda$ be a bound on the number of rounds available to complete the
process of influencing all nodes of the network.
The Time-Bounded Targeting with Incentives problem is 
to find incentives of minimum
cost which result in all nodes being influenced in at
most $\lambda$ rounds:

\smallskip

{\sc Time-Bounded Targeting with Incentives (TBI)}.

{\bf Instance:} A network $G=(V,E)$ with thresholds $t:V\longrightarrow \mathbb{N}$ and time bound $\lambda$.

{\bf Problem:} Find  incentives $p:V\longrightarrow \N_0$ of minimum
cost $\sum_{v\in V} p(v)$  s.t. $\Infl[p,\lambda] =V$.

\begin{figure}[tb!]
	\centering
		\includegraphics[width=\textwidth]{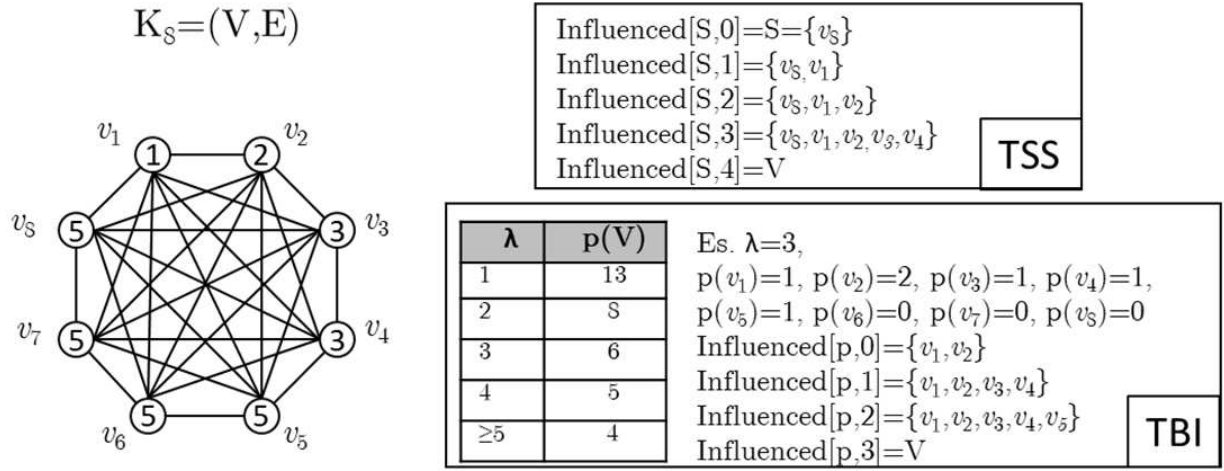}
		\caption{A complete graph $K_8$. The number inside each circle is the node threshold. Optimal solutions for the TSS problem and the TBI problem, with various values of $\lambda$, are shown. \label{example1} }
\end{figure}

\begin{example}
Solutions to the TBI problem can be quite different from solutions to
the TSS problem for a given network. 
Consider a complete graph $K_8$ on $8$ nodes with thresholds shown in  Fig. \ref{example1}. 
The optimal target set is $S{=}\{v_8\}$ which results in all nodes being influenced in $4$ rounds. 
The TBI problem admits different optimal solutions (with different incentive functions) depending on the value of $\lambda$, as shown in Fig. \ref{example1}.  
\end{example}

\subsection{Related Work and Our results}

The study of the spread of influence  in complex networks has  experienced a surge of interest in the last few years \cite{BBC14,CWY09,FGH12,GBL11,GLL11,GBLV13,HJBC14,LBGL13}. 
The algorithmic question of choosing the target set of size $k$ that activates the most number of nodes in the context of viral marketing was first posed by Domingos and Richardson \cite{DR-01}.
Kempe \emph{et al.} \cite{KKT03} started the study of this problem as a discrete optimization problem, and studied it in both the probabilistic independent cascade model and the threshold model of the influence diffusion process. They showed the NP-hardness of the problem in both models, and showed that a natural greedy strategy has a $(1 - 1/e - \epsilon)$-approximation guarantee in both models; these results were generalized to a more general cascade model in \cite{KKT05}.
However, they were mostly interested in networks with  randomly chosen thresholds.

In the  TSS problem, the size of the target set is not specified in advance, but the goal is to activate the entire network. 
Chen \cite{Chen-09} studied the TSS problem.
He proved  a strong   inapproximability result that makes unlikely the existence
of an  algorithm with  approximation factor better than  $O(2^{\log^{1-\epsilon }|V|})$.
Chen's result stimulated a series of papers including
\cite{ABW-10,BCNS,BHLN11,Centeno12,Chang,Chun,Chun2,Chopin-12,Cic14,Cic+,C-OFKR,CGRV17,CGR18,CGM+18,Fr+,GHPV13,Gu+,Li+,Mo+,NNUW,Re,W+,Za} 
that isolated
many interesting scenarios
in which the problem (and variants thereof) become tractable.
Ben-Zwi et al. \cite{BHLN11} generalized Chen's result on trees to show  that target set selection can be solved in $n^{O(w)}$ time where $w$
 is the treewidth of the input graph. The effect of several parameters,  such as diameter and  vertex cover number,
 of the input graph on the complexity of the problem are studied in \cite{NNUW}. The Minimum Target Set has also been studied 
from the point of view of the spread of disease or epidemics. For example, 
in \cite{DR09}, the case when all nodes have a threshold $k$ is studied; the authors showed that the problem is NP-complete for fixed $k \geq 3$.

The problem of 
maximizing the number of nodes activated within a specified number of rounds has also been
studied \cite{Cic14,Cic+,DZNT14,Li+,LP14}. The problem of dynamos or dynamic monopolies in graphs  is essentially the target set problem with every node  threshold being half its degree \cite{Peleg02}. The recent monograph \cite{CLC} contains an excellent overview of the area.

A problem similar to our work,  but considering the independent cascade model, has been considered in \cite{DAngelo17a,DAngelo17b}.

The Influence Maximization problem
with  incentives was introduced in \cite{Dem14}. In this model the authors assume that the thresholds are randomly chosen values in $[0,1]$ and they aim to understand how a fractional version of the Influence Maximization problem differs from the original version.
{To that purpose}, they introduced the concept of partial influence and showed that, in theory, the fractional version retains essentially the same computational hardness as the integral version, but in practice, better  solutions can be computed using heuristics in the fractional setting.

The Targeting with Partial Incentives (TPI) problem, of  finding incentives $p:V\longrightarrow \N_0$ of minimum
cost $\sum_{v\in V} p(v)$  such that all nodes are eventually influenced,
was studied in \cite{CGRV15}.
Exact solutions to the TPI problem for special classes of graphs were proposed in  \cite{CGRV15,CGR16,ER18a,ER18b}. 
Variants of the problem, in which the incentives are  modelled as additional links from an external entity, were studied in \cite{LATA2,LATA1,CGL+18}. The authors of  \cite{Li+} study the case in which 
 offering discounts to nodes causes them to be influenced with a probability proportional to the amount of the discount.

It was shown in \cite{CGRV15}  that  the TPI problem 
 cannot be approximated to within a ratio 
of  $O(2^{\log^{1-\epsilon} n})$,  for any fixed $\epsilon>0$, 
unless $NP\subseteq DTIME(n^{polylog(n)})$, where $n$ is the number of nodes in the graph. 
 As a consequence,  
for general graphs, the same inapproximability result still holds for the time bounded version of the problem that we study in this paper.
\begin{theorem}
The TBI problem 
 cannot be approximated to within a ratio 
of  $O(2^{\log^{1-\epsilon} n})$,  for any fixed $\epsilon>0$, 
unless $NP\subseteq DTIME(n^{polylog(n)})$, where $n$ is the number of nodes in the graph. 
\end{theorem}

\noindent
{\bf Our Results.} Our main contributions are polynomial-time algorithms for path, complete,
and tree networks.
In Section~\ref{sec:paths}, we present a linear-time greedy algorithm to allocate 
incentives to the nodes of a path network.
In Section~\ref{sec:complete}, we design a $O(\lambda n\log n)$ dynamic programming algorithm to
allocate incentives to the nodes of a complete network.
In Section~\ref{sec:trees}, we give an $O(\lambda^2 \Delta n) $algorithm to 
allocate incentives to  a tree 
with $n$ nodes and  maximum degree $\Delta$.


\section{A Linear-Time Algorithm for Paths}\label{sec:paths}

\vspace{-0.1truecm}
In this section, we present a greedy algorithm to allocate  incentives
to nodes of a path network.
We prove that our algorithm is linear-time.
%
%
\newcommand{\jk}[2]{L(#1,#2)}

\begin{figure}[tb!]
	\centering
		\includegraphics[width=0.85 \textwidth]{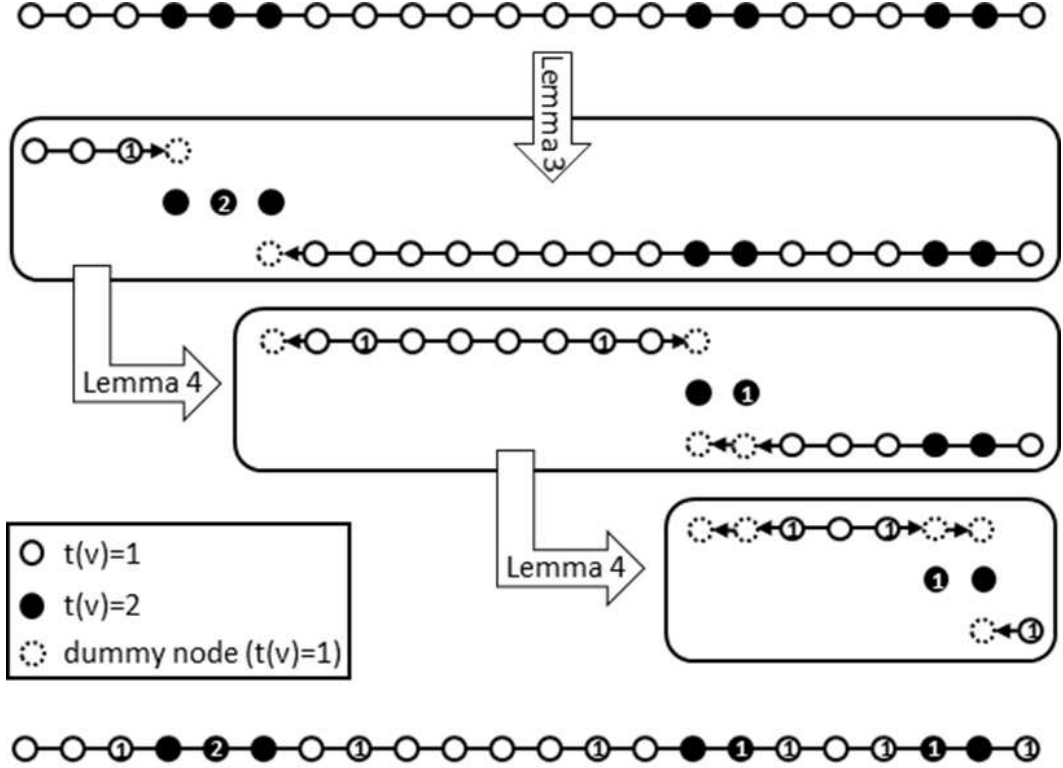}
		\caption{An example of the execution of the Algorithm \ref{alg1} on a path $\jk{0}{21}$ with a 2-path satisfying Lemma \ref{lemmaSub2} and two 2-paths satisfying Lemma \ref{lemmaTwoTwo}. Filled nodes represents nodes having threshold 2. Dashed nodes represents dummy nodes. The number inside the nodes represents the incentive assigned to the node.\label{fig1}} 
\end{figure}

We denote by  $\jk{0}{n-1}$ the  path with $n$ nodes $0, \ldots , n-1$ and edges $\{(i, i+1): 0 \leq i \leq n-2\}$. 
Since the threshold of each node cannot exceed its degree, we have that $t(0) = t(n-1) = 1$
and $t(i) \in \{1, 2\}$, for  $i = 1, \ldots , n {-} 2$.
For $0\leq j\leq k\leq n-1$, we denote by   $\jk{j}{k}$  the subpath induced by the nodes $j,\ldots, k$.

\begin{lemma} \label{lemmaLB}
Let $\jk{j}{k}$ be a subpath of $\jk{0}{n-1}$ with $t(j+1)=\cdots=t(k-1)=2$ and $t(j)=t(k)=1$.  For any incentive function $p:V \to \{0,1,2\}$ that solves the TBI problem on $\jk{0}{n-1}$ and for any $\lambda$,

$$\sum_{i=j+1}^{k-1} p(i) \geq \begin{cases}
k{-}j{-}2  & \text{if both 
                    $j+1$ and $k-1$ are influenced by  $j$ and $k$, resp. }  \\
k{-}j{-}1  &  \text{if either $j+1$ or $k-1$ is influenced by its neighbour ($j$ or $k$)} \\
k-j & \text{otherwise.}
\end{cases}$$ 
\end{lemma}

\proof
Let $p$ be an incentive function that solves the TBI problem on $\jk{0}{n-1}$.
For any node $i\in \{j+1,\ldots,k-1\}$, let $inf(i)\in\{0,1,2\}$ be the amount of influence that $i$ receives 
from its neighbours in $\jk{0}{n-1}$ during the influence process starting with 
$p$ {(that is, the number of i's neighbours that are influenced before round $i$)}. . 

For each $i=j+1,\ldots, k-1$, it must hold that $inf(i)+p(i) \geq t(i)=2$.
Hence,
\begin{equation}\label{eq1}
\sum_{i=j+1}^{k-1} p(i) \geq \sum_{i=j+1}^{k-1} (2-inf(i)) \geq     2(k-j-1) - \sum_{i=j+1}^{k-1} inf(i).
\end{equation}
Noticing that each link in $E$ is used to transmit  influence in at most one direction, we have

\begin{equation*}\label{eq2}
\sum_{i=j+1}^{k-1} inf(i) \leq \begin{cases}
k{-}j,  & \text{if both 
                    $j+1$ and $k-1$ are influenced by $j$ and $k$, resp.}  \\
k{-}j{-}1,  & \text{if either $j{+}1$ or $k{-}1$ is influenced by its neighbour ($j$ or $k$)} \\
k{-}j{-}2, & \text{otherwise.}
\end{cases}
\end{equation*}

As a consequence, using  equation  (\ref{eq1}) gives the desired result.

\newcommand{\mytodo}[1]{\textcolor{red}{Check: #1}}

\newcommand{\jko}[2]{OPT(#1,#2)}
\newcommand{\jkdxl}[2]{OPT(#1,#2,\leftarrow)}
\newcommand{\jksxr}[2]{OPT(\rightarrow, #1,#2)}

\newcommand{\jkdxr}[2]{OPT(#1,#2,{\substack{1\\ \rightarrow}})}
 \newcommand{\jksxl}[2]{OPT({\substack{1\\ \leftarrow},#1,#2)}}
\newcommand{\jksxll}[3]{OPT({\substack{#3\\ \leftarrow},#1,#2)}}
\newcommand{\jkdxrr}[3]{OPT(#1,#2,{\substack{#3\\ \rightarrow}})}

In the following we assume that $\lambda\geq 2$. The  case $\lambda=1$ will follow from the results  in Section \ref{sec:trees}, since the algorithm for trees has linear time when both $\lambda$ and the maximum degree are constant. 
\begin{definition} \label{def1}
We denote by $\jko{0}{n-1}$ the value of an optimal solution $p:V \to \{0,1,2\}$ to the TBI problem on  $\jk{0}{n-1}$ in $\lambda$ rounds.
For any subpath $\jk{j}{k}$ of $\jk{0}{n-1}$, we denote by: 
\begin{itemize}

\item[i)] $\jko{j}{k}$ the value $\sum_{i=j}^k p(i)$ where $p$ is  an  optimal solution to the TBI problem on $\jk{j}{k}$;

	\item[ii)] $\jkdxl{j}{k}$ the value $\sum_{i=j}^k p(i)$ where $p$ is  an  optimal solution to the TBI problem on $\jk{j}{k}$ with the additional condition that  the node $k$ gets one unit of influence   from   $k+1$;

	\item[iii)] $\jkdxrr{j}{k}{\ell}$ the value $\sum_{i=j}^k p(i)$ where $p$ is an  optimal solution to the TBI  problem on $\jk{j}{k}$ with the additional condition that    $k$ is influenced by round $\lambda-\ell$ without 
	getting influence from node  $k+1$;

	\item[iv)] $\jksxr{j}{k}$ the value $\sum_{i=j}^k p(i)$ where $p$ is an   optimal solution to the TBI problem  on $\jk{j}{k}$ with the additional condition that  $j$ gets one unit of influence   from   $j-1$;
		
	\item[v)] $\jksxll{j}{k}{\ell}$ the value $\sum_{i=j}^k p(i)$ where $p$ is an  optimal solution to the TBI  problem on $\jk{j}{k}$ with the additional condition that   node $j$ is influenced by round $\lambda-\ell$ without getting influence from  $j-1$. 
\end{itemize}
\end{definition}

\begin{lemma}\label{lemmaRes}
For any subpath $L(j,k)$ and for each $1\leq \ell<\ell'\leq \lambda$:

\smallskip

(1) If $t(k)=1$ then
 {\small $\jkdxl{j}{k} \leq \jko{j}{k}  \leq \jkdxrr{j}{k}{\ell} \leq  \jkdxrr{j}{k}{\ell'} \leq \jkdxl{j}{k}{+}1$.}
 
\smallskip

(2) If $t(j)=1$ then
 {\small
$\jksxr{j}{k} \leq \jko{j}{k}   \leq \jksxll{j}{k}{\ell} \leq \jksxll{j}{k}{\ell'} \leq \jksxr{j}{k}{+}1$.}
\end{lemma}
\begin{proof}
We first prove (1).
We notice that each of the first three inequalities\\
\noindent $\jkdxl{j}{k} {\leq} \jko{j}{k}, \   \jko{j}{k}{\leq} \jkdxrr{j}{k}{\ell},\ \jkdxrr{j}{k}{\ell} {\leq}  \jkdxrr{j}{k}{\ell'}$

\smallskip\noindent
 is trivially true since each solution that satisfies the assumptions of the right term is also a solution that satisfies the assumptions of the left term.
 It remains to show that $\jkdxrr{j}{k}{\ell'} \leq \jkdxl{j}{k} +1.$
	Let $p$ be a solution that gives  $\jkdxl{j}{k}$.
	Consider  $p'$ such that 
	$p'(i)=p(i)$, for each $i=j,\ldots,k-1$ and $p'(k) =1.$ 
Recalling that $t(k)=1$, we get that the cost   increases by at most 1 and 
 $p'$ is  a solution 
	in which   node $k$ is influenced at round $0\leq \lambda-\ell'$. A similar  proof holds for (2).
	\qed
\end{proof}

\begin{definition}\label{2path}
 $\jk{j}{k}$, with $j+1\leq k-1$,   is called a 2-path if  $t(j+1)=\ldots=t(k-1)=2$ and $t(j)=t(k)=1$.
\end{definition}

\begin{lemma}\label{lemmaSub2}
For any value of $\lambda$, if  $\jk{j}{k}$ is a 2-path with $m=k-j-1\neq 2$ then 
$$\jko{0}{n-1}= \jkdxr{0}{j}+k-j-2 +  \jksxl{k}{n-1}.$$
\end{lemma}

\begin{proof}
The links $(j,j+1)$ and $(k-1,k)$ can be used to transmit influence in at most one direction. This and    Lemma \ref{lemmaLB} imply that there exists two integers $1\leq \ell, \ell'\leq\lambda$ such that  $\jko{0}{n-1}$ is at least: 
\begin{enumerate}
\item
$\jkdxrr{0}{j}{\ell}+ k-j-2 +  \jksxll{k}{n-1}{\ell'}$,\\
\hphantom{aaaaaaaaaaaaaaaaaaaaaaaa}{if  both $j+1$ and $k-1$ obtain  influence from $j$ and $k$, respectively;}
\item
$\jkdxl{0}{j}+ k-j-1 +  \jksxll{k}{n-1}{\ell'}$,\\
\hphantom{aaaaaaaaaaaaaaaaaaaaaaaa}{ if   $k-1$  obtains  influence from $k$ but $j+1$ does not from $j$;}
\item
$\jkdxrr{0}{j}{\ell}+ k-j-1 +  \jksxr{k}{n-1}$,  \\
\hphantom{aaaaaaaaaaaaaaaaaaaaaaaa}{if   $j+1$  obtains  influence from $j$ but $k-1$ does not from $k$;}
\item
$\jkdxl{0}{j}+ k-j +\jksxr{k}{n-1}$\\
\hphantom{aaaaaaaaaaaaaaaaaaaaaaaa}{ if   neither $k-1$  obtains  influence from $k$ nor $j+1$ does  from $j$.}
\end{enumerate}
By (1) and (2) of Lemma \ref{lemmaRes}, we know that

\smallskip
\centerline{$\jkdxrr{0}{j}{\ell} \leq \jkdxl{0}{j} +1 \quad \mbox{ and }\quad  \jksxll{k}{n-1}{\ell'} \leq \jksxr{k}{n-1}+1$.}

\smallskip
\noindent
As a consequence, we get that in each of the above cases 1-4

\smallskip
\centerline{$\jko{0}{n-1} \geq \jkdxrr{0}{j}{\ell}+  k-j-2+\jksxll{k}{n{-}1}{\ell'}.$}

\smallskip
\noindent
From which, applying again  Lemma \ref{lemmaRes}, we get  
\begin{equation}\label{eqcan}
 \jko{0}{n-1}  \geq\jkdxr{0}{j}+  k-j-2+\jksxl{k}{n{-}1}.\end{equation}
We show now that the bound in (\ref{eqcan}) can be reached. Let $m=k{-}j{-}1$ be the number of nodes in $\jk{j{+}1}{k{-}1}.$
\\
If $m$ is odd, we can use the sequence of incentives $0(20)^*$ in order to influence  the nodes in $\jk{j{+}1}{k{-}1}$, provided that both $j+1$ and $k-1$ are influenced by $j$ and $k$ respectively within round $\lambda -1$ (see Fig. \ref{fig1}).\\
If $m$ is even, recalling that $\lambda >1$,   we can use the sequence of incentives $01(20)^*$. \\
By noticing that the cost of any  of the two sequences   is $k{-}j{-}2$, we get the desired result.
\qed

\end{proof}

\begin{lemma} \label{lemmaTwoTwo}
For any time bound $\lambda>1$, if  $t(0)=t(1)=\ldots =t(j-1)=1$ and $\jk{j}{j+3}$ is a 2-path then
$\jko{0}{n{-}1}$ is equal to

\smallskip

\centerline{$1{+}\min\Big\{\jkdxr{0}{j}{+}\jksxll{j{+}3}{n{-}1}{2},
\jkdxrr{0}{j}{2}{+}\jksxl{j{+}3}{n{-}1}\Big\}.$}
\end{lemma}
\begin{proof}{}
Let $p$ any solution for $L(0,n-1)$.
We consider all the possible ways in which the nodes $j+1$ and $j+2$ can get incentives
\begin{enumerate}
\item $p(j+1)=0, p(j+2)=1$. In order to influence  $j+1$ by round $\lambda$ the nodes $j$ and  $j+2$ must be influenced by round $\lambda-1$  which in turn needs 
	$j+3$  be influenced by round $\lambda-2$.
	Consequently, 
		\begin{equation}\sum_{i=0}^{n-1}p(i)\geq  1+ \jkdxr{0}{j}+\jksxll{j+3}{n{-}1}{2}.\label{eqc2}\end{equation}
\item $p(j+1)=1, p(j+2)=0$. As in  case 1, we have 	
	\begin{equation}
	\sum_{i=0}^{n-1}p(i)\geq 1+ \jkdxrr{0}{j}{2}+\jksxl{j+3}{n{-}1}. \label{eqc3}
	\end{equation}
\item $p(j+1)=p(j+2)=0$.  This case cannot hold since  none of the two nodes could  be influenced (they have threshold $2$ and can get at most one unit of influence from one neighbour).  
	\item $p(j+1)=1, p(j+2)=1$.  In this case,  either  $j+1$ receives one unit of influence from $j$ or  $j+2$ receives one unit of influence from $j+3$ by some time $\leq \lambda -1$.  Assume first that  there is an optimal solution in which $j+1$ receives one unit of influence from $j$ at round  $\lambda -\ell$. This means that
	$$\sum_{i=0}^{n-1}p(i)\geq \jkdxrr{0}{j}{\ell} + 2 + 
	\begin{cases}{\jksxr{j+3}{n-1}} &{\mbox{if $\ell>2$}}\\
	              {OPT(j+3,n-1) } &{\mbox{if $\ell\leq2$}}\end{cases}$$
By Lemma \ref{lemmaRes},   
	it holds $OPT(j+3,n-1) \geq \jksxr{j+3}{n-1}$ and $\jkdxr{0}{j} \leq \jkdxrr{0}{j}{\ell}$, for any $1< \ell \leq \lambda$; hence, 
	$$\sum_{i=0}^{n-1}p(i)\geq \jkdxrr{0}{j}{1} + 2 + \jksxr{j+3}{n-1} $$
Moreover, by Lemma \ref{lemmaRes}, we know that $	\jksxr{j+3}{n-1}\geq \quad \jksxll{j+3}{n-1}{2} -1$
	and we get again the bound  (\ref{eqc2}).
	
	Similarly, if we assume  that  there is an optimal solution in which $j+2$ receives one unit of influence from $j+3$ by some time $\leq \lambda -1$, we can obtain that 
	the bound   (\ref{eqc3}) holds.
	\item $p(j+1)=2, p(j+2)=0$.   
	In this case, recalling that for any $1< \ell \leq \lambda,$ $\jksxl{j{+}3}{n-1} \leq \jksxll{j{+}3}{n{-}1}{\ell}$ we have  
	$$	\sum_{i=0}^{n{-}1}p(i)\geq  2+  	\jkdxl{0}{j}+\jksxl{j{+}3}{n{-}1} 
	\geq 	1+  \jkdxrr{0}{j}{2} +\jksxl{j{+}3}{n{-}1} 
$$
	where the last inequality is due to   (1) of Lemma \ref{lemmaRes}, e.g. $\jkdxrr{0}{j}{2}\leq  \jkdxl{0}{j}+1$.  Hence, the bound (\ref{eqc3}) holds in this case.
	\item $p(j+1)=0, p(j+2)=2$.  Similarly as in case 5, one gets the bound (\ref{eqc2}).
		\item $p(j+1)+p(j+3)\geq 3$. In this case, we have  
	$$	\sum_{i=0}^{n-1}p(i)\geq  3+  	\jkdxl{0}{j}+\jksxr{j+3}{n{-}1}.$$
	By (1) and (2) of Lemma \ref{lemmaRes}, we get  that no such a solution can   beat the bound in  (\ref{eqc2}).
\end{enumerate}
Overall,    we can always find an optimal solution satisfying either Case $1$ or $2$  and the Lemma follows.
\qed
\end{proof}

%

\begin{lemma}\label{opt1}
For any value of $\lambda$,  the minimum cost for the TBI problem on a path   of $n$ nodes having threshold $1$ is $\left \lceil n/(2\lambda+1)\right \rceil$.
\end{lemma}

\begin{proof}{}
For any value of $\lambda$, we know that each incentive given to a node can be used to influence at most $2\lambda+1$ nodes within $\lambda$ rounds.

Consider a path of $n$ nodes having threshold $1$, a simple strategy that requires an optimal number (i.e., $\left \lceil n/(2\lambda+1)\right \rceil$) of incentives is the following.
If $n\leq 2\lambda+1,$ a single incentive to any middle node of the path is enough to influence the whole path in at most $\lambda$ rounds (i.e, $p(\lfloor n/2 \rfloor)=1$ and $p(i)=0$ for each $i\neq \lfloor n/2 \rfloor$). Otherwise, 
\begin{equation} \label{OptimalAssigment}p(i)=\begin{cases}1 & \text{ if } i=\lambda +c(2\lambda+1) \text{ for some } c\in\left \{0,\ldots,\left \lfloor \frac{n}{2\lambda+1}\right \rfloor-1\right \} \\ 
1 & \text{ if } i=n-1-\lambda \text  { and  } \left \lfloor \frac{n}{2\lambda+1}\right \rfloor \neq \left \lceil \frac{n}{2\lambda+1}\right \rceil \\0& \text{otherwise.} \end{cases}
\end{equation} \qed
\end{proof}


\begin{remark} \label{obsDummy}

 $\jkdxrr{j}{k}{\ell}$  can be obtained by solving the TBI  problem on an augmented path $\jk{j}{k+\ell}$ obtained from $\jk{j}{k}$ by concatenating   $\ell$  dummy nodes  on the right of $k$ with $t(k+1)=t(k+2)=\ldots=t(k+\ell)=1.$ 
Notice  that, for $\ell\leq\lambda$, it is always possible to find an optimal assignment of incentives for the augmented path $\jk{j}{k+\ell}$ in which all dummy nodes get incentive $0$.
Indeed it is possible to obtain such an assignment starting from any optimal assignment and moving the incentives from the dummy nodes to node $k$. 	
An analogous observation  holds for  $\jksxll{j}{k}{\ell}$.
\end{remark}

\begin{algorithm}[ht!]
\SetCommentSty{footnotesize}
\SetKwInput{KwData}{Input}
\SetKwInput{KwResult}{Output}
\DontPrintSemicolon
\caption{ \ \    TBI-Path($\jk{0}{n{-}1}$) \label{alg1}}
\KwData {A Path $\jk{0}{n{-}1}$, thresholds $t(i) \in \{1,2\}$, $i=0,\ldots,n-1$, and a time bound $\lambda$.}
\KwResult{A solution $p(i):V \rightarrow \{0,1,2\}$ of the TBI problem. }
\setcounter{AlgoLine}{0}
$i=0$\\
\While {there exists a node $j$ with $t(j)=2 \ $  for some $i<j<n-1$}{
Identify  the leftmost 2-path in the current path $L(i,n-1)$; let it be $L(j,k)$.\\
\If{$L(j,k)$ is a  2-path satisfying  Lemma \ref{lemmaSub2} }{ assign incentives to the nodes $j+1,\ldots,k-1$
as in  Lemma \ref{lemmaSub2};   \\
$t(j+1)=t(k-1)=1$;\\
obtain $p(i),\ldots, p(j)$ by using Lemma \ref{opt1} on $\jk{i}{j{+}1}$ with {\small $t(i)=\cdots=t(j{+}1)=1$}; \\
$i=k-1$;}
\ElseIf{$L(j,k=j+3)$ is a  2-path satisfying  Lemma \ref{lemmaTwoTwo}}{
\If(\tcp*[f]{Case 1 of Lemma \ref{lemmaTwoTwo} }){$j-i+2=c(2\lambda+1)$ for some $c > 0$}{
$p(j+1)=0;$ $\quad p(j+2)=1;$ $\quad i'=j+1;$
} 
\Else(\tcp*[f]{Case 2 of Lemma \ref{lemmaTwoTwo} }){
$p(j+1)=1;$ $\quad  p(j+2)=0;$ $\quad i'=j+2;$
}
$t(j+1)=t(j+2)=1$;\\
obtain $p(i),\ldots, p(j)$ by using Lemma \ref{opt1} on $\jk{i}{i'}$ with {\small $t(i)=\cdots=t(i')=1$}; \\
$i=i';$\\

}
}
Assign incentives to  $\jk{i}{n-1}$ (with $t(i)=\ldots=t(n-1)=1$), using Lemma \ref{opt1};\\
\Return $p$;
\end{algorithm}
\remove{The above results tell  that  whenever there are nodes of threshold 2, 
there is an optimal solution in which  a maximal length subpath of $m$ nodes having threshold $2$, receives the incentives as suggested by  Lemma \ref{lemmaSub2} (i.e., $0(20)^*$ when $m$ is odd and  $01(20)^*$ when $m>2$ is even) or Lemma  \ref{lemmaTwoTwo} (i.e., $01$ or $10$ when $m=2$).} 

\smallskip

{Our algorithm  iterates  from left to right, identifying all of the 2-paths and, using  Lemma~\ref{lemmaSub2} or~\ref{lemmaTwoTwo} and Lemma~\ref{opt1}, it  optimally assigns incentives both to the nodes of threshold 2 and to the nodes (of threshold $1$) on the left. It then  removes them from the original path. 
Eventually, it will  deal with a last subpath in which  all of the nodes have threshold $1$. 
} 
\begin{theorem}
For any time bound $\lambda>1$, Algorithm \ref{alg1} provides an optimal solution for the  TBI problem on any path $\jk{0}{n{-}1}$ in time $O(n)$.
\end{theorem}
\proof
We  show that the algorithm  selects an optimal strategy according to the length and the position of the leftmost 2-path $\jk{j}{k}$   and then iteratively operates on the subpath $\jk{i}{n-1}$ where  $i=k-1$ (one dummy node on the left) or $i=k-2$ (two dummy nodes on the left). See Fig. \ref{fig1}.

Let $\jk{i}{n-1}$ be the current path and $\jk{j}{k}$ be the leftmost 2-path. If $\jk{j}{k}$ satisfies the hypothesis of Lemma \ref{lemmaSub2}, then we have \\
\hspace*{1truecm}$ \jko{i}{n-1}= \jkdxr{i}{j}+k-j-2 +  \jksxl{k}{n-1}.$\\
Hence, we can obtain   optimal   incentives for nodes $ i, \ldots,j $  by using the result in Lemma \ref{opt1} on $\jk{i}{j+1}$ (where $j+1$ is a dummy node). Moreover, we  assign $k-j-2$ incentives to the nodes $j + 1,\ldots, k - 1$ as suggested  in Lemma \ref{lemmaSub2} (i.e., $0(20)^*$ when the length of the 2-path is odd and  $01(20)^*$ otherwise) and the algorithm iterates on $\jk{k-1}{n-1}$ (where $k-1$ is a dummy node).

Now suppose that $\jk{j}{k=j+3}$ satisfies the hypothesis of  Lemma \ref{lemmaTwoTwo}. We have that
$\jko{i}{n{-}1}$ is equal to
\begin{equation} \label{opteq}
1+\min\Big\{\jkdxr{i}{j}+\jksxll{k}{n{-}1}{2},\jkdxrr{i}{j}{2}+\jksxl{k}{n{-}1}\Big\}.
\end{equation}
We have two cases to consider, according to the distance between $i$ and $j$.

\smallskip
\noindent
First assume that  $j-i+2=c(2\lambda+1)$ for some $c > 0$. 
	By Lemma \ref{opt1} and Remark \ref{obsDummy} we know that in this case 
	$\jkdxrr{i}{j}{2}=\jkdxr{i}{j}+1$ and since by (2) of Lemma \ref{lemmaRes} we know that $\jksxll{k}{n-1}{2} \leq \jksxl{k}{n-1} +1$, we have that 
$\jkdxr{i}{j}+\jksxll{k}{n-1}{2}$ corresponds to the minimum of equation (\ref{opteq}) and hence the solution described by case 1 in Lemma \ref{lemmaTwoTwo} (i.e.,  $p(j+1) = 0,p(j+2) = 1$) is optimal.
Incentives to $  i, \ldots,j $ are assigned exploiting the result in Lemma \ref{opt1} on $\jk{i}{j+1}$ (where $j+1$ is a dummy node) and the algorithm iterates on $\jk{k-2}{n-1}$ (where both $k-1$ and $k-2$ are dummy nodes).

\smallskip
\noindent
Now assume that $j-i+2\neq c(2\lambda+1)$ for some $c > 0$. 
In this case, we have $\jkdxrr{i}{j}{2}=\jkdxr{i}{j}.$ \\
By (1) of Lemma \ref{lemmaRes} we know that
 $\jksxl{k}{n-1} \leq \jksxll{k}{n-1}{2}$. Hence,  $\jkdxrr{i}{j}{2}+\jksxl{k}{n-1}$ corresponds to the minimum of equation (\ref{opteq}) and  the solution in  case 2 in Lemma \ref{lemmaTwoTwo}  (i.e.,  $p(j+1) = 1,p(j+2) = 0$) is optimal.
Incentives to $  i, \ldots,j $ are assigned using the result in Lemma \ref{opt1}  on $\jk{i}{j+2}$ (considering both $j+1$ and $j+2$  as dummy nodes) and the algorithm iterates on $\jk{k-1}{n-1}$ (where $k-1$ is a dummy node).

\smallskip

If there remains  a last subpath of nodes of threshold one, this  is solved optimally using Lemma \ref{opt1}.

\noindent
\textbf{Complexity.} The identification of the 2-paths and their classification can be easily done in linear time. Then, the algorithm operates in a single pass from left to right and  the  time is $O(n)$.
\qed

\section{An $O(\lambda n \log n)$ Algorithm for Complete Graphs}\label{sec:complete}

In this section, we present an $O(\lambda n \log n)$ dynamic programming algorithm to allocate
incentives to the nodes of a complete network $K_n = (V,E)$.
We begin by proving that for any assignment of thresholds to the nodes of $K_n$,
there is an optimal solution in which the thresholds of
all nodes that are influenced {\em at} round $\ell$ are at least as large as
the thresholds of all nodes that are influenced before round $\ell$
for every $1 \leq \ell \leq \lambda$.

Let $K_m$ be the subgraph of $K_n$ that is induced by $V_m =\{v_1, v_2, \ldots, v_m\}$.
We will say that an incentive function $p:V_m\longrightarrow \N_0$ is
$\ell$\emph{-optimal for $K_m$}, $1\leq m \leq n$, $0 \leq \ell \leq \lambda$,
if $\sum_{v\in V_m} p(v)$ is the minimum cost to influence all nodes in
$V_m$ in $\ell$ rounds.

\begin{lemma}\label{sort-clique}
  Given $K_m$, thresholds $t(v_1) \leq t(v_2) \leq \ldots \leq t(v_m)$, and
  $1 \leq \ell \leq \lambda$,
  if there exists an $\ell$-optimal solution for $K_m$ that influences $k<m$
  nodes by the end of round $\ell-1$, then there is an $\ell$-optimal solution
  that influences $\{v_1, v_2, \ldots, v_k\}$ by the end of round $\ell-1$.
\end{lemma}

\remove{\begin{outline}
  The proof is an exchange argument that transforms an $\ell$-optimal incentive function $p^*$ into an $\ell$-optimal incentive function $p$ that influences $\{v_1, v_2, \ldots, v_k\}$ by the end of round $\ell-1$. \qed
\end{outline}}
\proof 
  Let $p^*$ be an $\ell$-optimal incentive function for $K_m$ that influences a set $V_k^*=\{u_1, u_2, \ldots, u_k\}$ of $k$ nodes of $K_m$ by the end of round $\ell-1$.
  We will show how to construct an $\ell$-optimal incentive function for $K_m$ that influences nodes $V_k=\{v_1, v_2, \ldots, v_k\}$ by the end of round $\ell-1$ where $t(v_1) \leq t(v_2) \leq \ldots \leq t(v_k)$ and $t(v_j)\geq t(v_k)$ for $j=k+1,k+2,\ldots,m$.

  Suppose that $p$ is an incentive function for $K_m$ that influences nodes $V_k=\{v_1, v_2, \ldots, v_k\}$ by the end of round $\ell-1$. If $V_k^*$ is different from $V_k$, then there is some $u_i \in V_k^* \backslash V_k$ and some $v_j \in V_k \backslash V_k^*$ such that $t(u_i)\geq t(v_j)$.
  Since $v_j$ is influenced \emph{at} round $\ell$ in the $\ell$-optimal solution $p^*$, it must require the influence of $t(v_j)-p^*(v_j)$ neighbours.
  (If it required the influence of fewer neighbours, then $p^*$ would not be $\ell$-optimal.)
  Note that $t(v_j)-p^*(v_j)\geq 0$.
  Similarly, $u_i$ requires the influence of $t(u_i)-p^*(u_i)\geq 0$ neighbours.
  Consider the set of nodes $V_k^* \cup \{v_j\} \backslash \{u_i\}$ and define $p$ as follows.
  Choose $p(v_j)$  and $p(u_i)$ as
   
	\smallskip
   \centerline{$ t(v_j)-p(v_j)=t(u_i)-p^*(u_i) \quad \mbox{ and }\quad  t(u_i)-p(u_i)=t(v_j)-p^*(v_j)$}
	
	\smallskip
	\noindent
  so that $v_j$ is influenced at the same round as $u_i$ was influenced in the $\ell$-optimal solution and  $u_i$ is influenced at round $\ell$.
  Set  $p(v)=p^*(v)$ for all other nodes in $K_m$.
  The difference in value between $p$ and $p^*$ is
  
	\smallskip
   \centerline{$
    p(v_j){+}p(u_i){-}p^*(v_j){-}p^*(u_i) = 0
$}
	
	\smallskip
	\noindent
	We can iterate until we find an $\ell$-optimal solution that influences $\{v_1, v_2, \ldots, v_k\}$ by the end of round $\ell-1.\  \Box$

\medskip
By Lemma~\ref{sort-clique}, our algorithm can first sort the nodes by
non-decreasing threshold value w.l.o.g.
The sorting can be done in $O(n)$ time using counting sort because
$1 \leq t(v) \leq n-1 = d(v)$ for all $v\in V$.
In the remainder of this section, we assume that $t(v_1) \leq t(v_2) \leq \ldots \leq t(v_n)$.

Let $Opt_{\ell}(m)$ denote the value of an $\ell$-optimal solution for $K_m$,
$1\leq m \leq n$, $0 \leq \ell \leq \lambda$.
Any node $v$ that is influenced at round 0 requires incentive $p(v)=t(v)$ and
it follows easily that

\begin{equation}\label{clique-column0}
  Opt_0(m)=\sum_{i=1}^m t(v_i),\; 1\leq m \leq n.
\end{equation}

\noindent
Now consider a value $Opt_{\ell}(m)$ for some $1\leq m \leq n$ and $1 \leq \ell \leq \lambda$.
If exactly $j$ nodes, $1\leq j\leq m$, are influenced by the end of round $\ell-1$ in an $\ell$-optimal solution for $K_m$, then each of the $m-j$ remaining nodes in $V_m$ has $j$ influenced neighbours at the beginning of round $\ell$ and these neighbours are $v_1, v_2, \ldots, v_j$ by Lemma~\ref{sort-clique}.
For such a remaining node $v$ to be influenced at round $\ell$, either
$t(v)\leq j$ or $v$ has an incentive $p(v)$ such that $t(v)-p(v)\leq j$.
It follows that

\noindent
\begin{equation}\label{clique-recurrence}
  Opt_{\ell}(m)=\min_{1\leq j\leq m} \Bigl\{Opt_{\ell-1}(j) + \sum_{i=j+1}^m \max\{0,t(v_i)-j\}\Bigr\},\; 1\leq m \leq n.
\end{equation}

We will use $Ind_{\ell}(m)$ to denote the index that gives the optimal value $Opt_{\ell}(m)$, that is,

\noindent
\begin{equation}\label{index-clique-recurrence}
  Ind_{\ell}(m) = \argmin_{1 \leq j \leq m} \Bigl\{Opt_{\ell-1}(j) + \sum_{i=j+1}^m \max\{0,t(v_i)-j\}\Bigr\},\; 1\leq m \leq n.
\end{equation}

A dynamic programming algorithm that directly implements the recurrence in equations~(\ref{clique-column0}) and~(\ref{clique-recurrence}) will produce the optimal solution value $Opt_{\lambda}(n)$ in time $O(\lambda n^3)$.
We can reduce the complexity by taking advantage of some structural properties.

\begin{lemma}\label{clique-property4}
  For any $1 \leq \ell \leq \lambda$, if $k < m$ then $Ind_{\ell}(k) \leq Ind_{\ell}(m)$, $1 \leq k\leq n-1$, $2 \leq m\leq n$.
\end{lemma}
\begin{proof}{}
If $k < Ind_{\ell}(m)$, the lemma trivially holds because $Ind_{\ell}(k) \leq k$.
Now assume that $k \geq Ind_{\ell}(m)$ and suppose that $Ind_{\ell}(k)  > Ind_{\ell}(m)$ contrary to the statement of the lemma.
To simplify notation, let $\k=Ind_{\ell}(k)$ and $\m=Ind_{\ell}(m)$.
Then
\begin{eqnarray} \nonumber
    Opt_{\ell}(m) &=& Opt_{\ell-1}(\m) + \sum_{i=\m+1}^{m} \max\{0,t(v_i)-\m\}\\ \nonumber
		&=& Opt_{\ell-1}(\m) + \sum_{i=\m+1}^{k} \max\{0,t(v_i)-\m\} +	
		\sum_{i=k+1}^{m} \max\{0,t(v_i)-\m\}\\ \label{neq3}
		&\geq& Opt_{\ell-1}(\k) + \sum_{i=\k+1}^{k} \max\{0,t(v_i)-\k\}+
		\sum_{i=k+1}^{m} \max\{0,t(v_i)-\m\}  \\  \label{neq2}
		&>& Opt_{\ell-1}(\k) + \sum_{i=\k+1}^{k} \max\{0,t(v_i)-\k\}+
		\sum_{i=k+1}^{m} \max\{0,t(v_i)-\k\}  \\ \nonumber
    &=& Opt_{\ell-1}(\k) + \sum_{i=\k+1}^{m} \max\{0,t(v_i)-\k\}
  \end{eqnarray}
	 The inequality (\ref{neq3}) follows because $\k=Ind_{\ell}(k)$, and
 the inequality (\ref{neq2}) follows because 
$\max\{0,t(v_i)-\m\} > \max\{0,t(v_i)-\k\}$ for each $k+1 \leq i \leq m$.			
  This is a contradiction because $\m=Ind_{\ell}(m)$.
\qed
\end{proof}

\begin{theorem}\label{thm:clique}
  For any complete network $K_n=(V,E)$, threshold function
  $t:V\longrightarrow \N$, and $\lambda \geq 1$, the TBI problem can be solved
  in time $O(\lambda n \log n)$.
\end{theorem}
\begin{proof}
Our dynamic programming algorithm computes two $n \times (\lambda+1)$ arrays
$\VAL$ and $\IND$ and returns a solution $p$ of $n$  incentives.
$\VAL[m,\ell]=Opt_{\ell}(m)$ is the value of an $\ell$-optimal solution for $K_m$ (for a given threshold function $t: V \longrightarrow \N$), and $\IND[m,\ell]=Ind_{\ell}(m)$ is the index that gives the optimal value, $1\leq m \leq n$, $0 \leq \ell \leq \lambda$.

The array entries are computed column-wise starting with column 0.
The entries in column $\VAL[*,0]$ are  sums of thresholds according to~(\ref{clique-column0}) and the indices in $\IND[*,0]$ are all $0$, so these columns can be computed in time $O(n)$.
In particular,  $\VAL[j,0]=\sum_{i=1}^j t(v_i)$, $j=1,\ldots,m$.

Suppose that columns $1,2,\ldots, \ell-1$ of $\VAL$ and $\IND$ have been computed according to~(\ref{clique-recurrence}) and~(\ref{index-clique-recurrence}) and consider the computation of column $\ell$ of the two arrays.
To compute $\IND[m,\ell]$ for some fixed $m$, $1\leq m \leq n$, we define a function

\centerline{$A(j)=Opt_{\ell-1}(j) + \sum_{i=j+1}^m \max\{0,t(v_i)-j\}, \quad 1\leq j\leq m$}

\smallskip
\noindent
and show how to compute each $A(j)$ in $O(1)$ time.\\
 By (\ref{index-clique-recurrence}),  $\,\displaystyle Ind_{\ell}(m)= \argmin\{A(j)\,|\, 1 \leq j \leq m\}$.

First we compute an auxiliary vector $a$ where $a[j]$ contains the smallest integer $i \geq 1$ such that $t(v_i)\geq j$, $1 \leq j \leq n$.
This vector can be precomputed once in $O(n)$ time because the nodes are sorted by non-decreasing threshold value.
Furthermore, the vector $a$ together with the entries in column $\VAL[*,0]$ allow the computation of $\sum_{i=j+1}^{m} \max\{0,t(v_i)-j\}$ in $O(1)$ time for each $1 \leq j \leq n$.
Since $Opt_{\ell-1}(j)=\VAL[j,\ell-1]$ has already been computed, we can compute $A(j)$ in $O(1)$ time.
The values $Opt_{\ell}(m)=\VAL[m,\ell]$ can also be computed in $O(1)$ time for each $1 \leq m \leq n$ given $Ind_{\ell}(m)= \IND[m,\ell]$, vector $a$, and column $\VAL[*,0]$.
The total cost so far is $O(\lambda n)$.
It remains to show how to compute each column $\IND[*,\ell]$ efficiently.

\remove{
\begin{eqnarray*}
\sum_{i=j+1}^{m} \max\{0,t(v_i)-j\} 
&= &\sum_{i=a[j]+1}^{m} t(v_i)-j \\
&= & \sum_{i=a[j]+1}^{m} t(v_i) - \sum_{i=a[j]+1}^{m} j \\
&= & \sum_{i=a[j]+1}^{m} t(v_i) -  j (m-a[j]) \\
&= & \sum_{i=1}^{m} t(v_i) - \sum_{i=1}^{a[j]} t(v_i) -  j (m-a[j])\\
&= & b[m] - b[a[j]] - j (m-a[j])
\end{eqnarray*}
}

The following algorithm recursively computes the column $\IND[m,\ell]$, $1 \leq m \leq n$ assuming that columns $0,1,2, \ldots, \ell-1$ of  $\IND$ and $\VAL$ have already been computed. The algorithm also assumes that two dummy rows have been added to array $\IND$ with $\IND[0,\ell]=1$ and $\IND[n+1,\ell]=n$, $0 \leq \ell \leq \lambda$, to simplify the pseudocode.
The initial call of the algorithm is COMPUTE-INDEX($1,n$).

We claim that algorithm COMPUTE-INDEX($1, n$) correctly computes the values 
$\IND[m,\ell]$  for $1 \leq m \leq n$.
First, it can be proved by induction that when 
we call COMPUTE-INDEX($x,y$), the indices $\IND[x-1,\ell]$ 
and $\IND[y+1,\ell]$ have already been correctly computed.
By Lemma~\ref{clique-property4}, $Ind_{\ell}(x-1) \leq Ind_{\ell}(\lceil\frac{x+y}{2}\rceil) \leq Ind_{\ell}(y+1)$, so the algorithm correctly searches for $\IND[m,\ell]$ between $\IND[x-1,\ell]$ and $\IND[y+1,\ell]$.
 
\begin{algorithm}
\SetCommentSty{footnotesize}
\SetKwInput{KwData}{Input}
\SetKwInput{KwResult}{Output}
\DontPrintSemicolon
\caption{ \ \    COMPUTE-INDEX($x,y$) \label{alg2}}
\KwData {Indices $x,y$. } 
\KwResult{The values $\IND[i,\ell]$ for $i=x, \ldots y$.}
\setcounter{AlgoLine}{0}
\If(\tcp*[f]{Assume that $\IND[0,\ell]=1$ and $\IND[n+1,\ell]=n$}){$x \leq y$}{
$m=\lceil\frac{x+y}{2}\rceil$;\\
$\IND[m,\ell]= \argmin\left\{A(j)\, |\ \IND[x-1,\ell] \leq j \leq \min\{\IND[y+1,\ell],m\}\right\} $;\\
COMPUTE-INDEX($x,m-1$);\\
COMPUTE-INDEX($m+1,y$);} 
\end{algorithm}
 
It is not hard to see that the height of the recursion tree  obtained calling COMPUTE-INDEX($1, n$) is $\lceil \log (n+1)\rceil$.
Furthermore, the number of values $A(j)$ computed at each level of the recursion tree is $O(n)$ because the ranges of the searches in line 3 of the algorithm do not overlap (except possibly the endpoints of two consecutive ranges) by Lemma~\ref{clique-property4}. Thus, the computation time at each level is $O(n)$, and the computation time for each column $\ell$ is $O(n\log n)$.
After all columns of $\VAL$ and $\IND$ have been computed, the value of the optimal solution will be in $\VAL[n,\lambda]$.
The round during which each node is influenced and the optimal function $p$ of  incentives can then be computed by backtracking through the array $\IND$ in time $O(\lambda +n)$.
The total complexity is $O(\lambda n \log n)$.
\qed
\end{proof}
\section{A Polynomial-Time Algorithm for Trees} \label{sec:trees}

In this section, we give an algorithm for the TBI problem on trees.
Let $T = (V,E)$ be a tree having $n$ nodes
and  the maximum  degree $\Delta$ .
We will assume that $T$ is rooted at some node $r$.
Once such a rooting is fixed, for any node $v$, we  denote by $T_v$ the  subtree rooted at $v$,
and by $C(v)$ the set of children of $v$.
We will develop a dynamic programming algorithm that will  prove the following theorem.
\begin{theorem}\label{theorem-tree}
For any $\lambda>1,$ the TBI problem can be
solved in  time $O(n\lambda^2 \Delta)$ on  a tree having $n$ nodes and  maximum degree $\Delta$.
\end{theorem}

 The rest of this section is devoted to the description and analysis of the
  algorithm that proves  Theorem \ref{theorem-tree}.
The algorithm   performs a post-order traversal of the tree $T$ so that each node is considered
after all of its children have been processed. For each node $v$, the algorithm solves some TBI problems on the subtree $T_v$, with some restrictions on the node $v$ regarding its threshold and the round during which it is influenced.
For instance, in order to compute some of these values
we will consider not only the original threshold $t(v)$ of $v$, but also
the reduced threshold $t'(v)=t(v)-1$ which simulates the influence of the parent node.
\def\MIS{P}
\begin{definition} \label{defi:MIS}
For each node $v\in V$,  integers $\ell \in\{0,1,\ldots,\lambda\}$, and 
$t\in \{t'(v),t(v) \}$, let us
denote by $\MIS[v,\ell,t]$   the minimum cost of influencing all of the nodes in $T_v$,
in at most $\lambda$ rounds,  assuming that\\
$\bullet$ the threshold of $v$ is $t$, and for every $u\in V(T_v)\setminus\{v\}$, the threshold of $u$ is $t(u)$;\\
$\bullet$  $v$ is influenced by round $\ell$ in $T_v$ and is  able to start influencing its neighbours by round $\ell +1$.\footnote{Notice that this does not exclude the case that
$v$ becomes an influenced node at some round $\ell' < \ell$.}
\end{definition}
\noindent{Formally the value of $\MIS[v,\ell,t]$ corresponds to\\		 
$\displaystyle \MIS[v,\ell,t]  = \min_{\substack{p: T_v  \to\N_0, \ \Infl_{T_v}[p,\lambda]= T_v \\ |C(v)   \cap \Infl_{F(v,d)}[p,\ell-1]| \geq t-p(v) }  } \left \{ \sum_{v\in T_v} p(v) \right\}$}
We set  $\MIS[v,\ell,t]= \infty$ when the above problem is infeasible.
{Denoting  by $p_{v,\ell,t}: V(T_v) \rightarrow \N_0$ the incentive function attaining the value $\MIS[v, \ell,t]$,} the parameter $\ell$ is such that:
\begin{enumerate}
	\item if $\ell=0$ then $p_{v,\ell,t}(v)=t$,
	\item otherwise, $v$'s  children can influence $v$ at round $\ell$, i.e. $|\{C(v)  \cap \Infl[p_{v,\ell,t},\ell-1]\}| \geq t-p_{v,\ell,t}(v)$.
\end{enumerate}

\begin{remark}
It is worthwhile mentioning that  $\MIS[v, \ell,t]$ is monotonically  non-decreasing in $t$.
However, $\MIS[v,\ell,t]$ is not necessarily  monotonic in $\ell$.
\end{remark}
{Indeed, partition the set $C(v)$ into two sets: $C'(v)$, which contains the $c$ children that influence $v$, and $C''(v)$,  which contains the remaining $|C(v)|-c$ children that may be influenced by $v$. 
A small value of $c$ may require a higher cost on subtrees rooted at a node $u \in C'(v)$, and may save some budget on the remaining subtrees; the opposite happens for a large value of $c$.}

The minimum cost  to influence  the nodes in $T$    in  $\lambda$ rounds  {follows from decomposing
the optimal solution according to the round on which the root is influenced and} 
can then be obtained by computing
\begin{equation}\label{eq-mas}
\min_{0\leq \ell \leq \lambda} \ \MIS[r,\ell,t(r)].
\end{equation}
We  proceed using a post-order traversal of the tree, so that the computations of the various values $\MIS[v,\ell,t]$ for a node $v$ are done after all of the values for $v$'s children are known.
For each leaf node $v$ we have

\begin{equation}\label{eq-casel}
\MIS[v, \ell,t] =  \begin{cases} 1 & \mbox{ if } \ell=0 \mbox{ and } t=t(v)=1
\\
0 & \mbox{ if } 1\leq \ell \leq \lambda \mbox{ and } t=t(v)-1=0
\\
\infty &  \mbox{otherwise.} \end{cases}
\end{equation}

\noindent
Indeed, a leaf $v$ with threshold $t(v)=1$ is influenced in the one-node subtree $T_v$ only when either  $p_{v,\ell,t}(v)=1$   ($\ell=0$),   or  for some $1\leq \ell \leq \lambda$,   it is influenced by its parent (i.e., the residual threshold $t=t(v)-1=0$).

For any internal node $v$, we show how to compute each  value $\MIS[v,\ell,t]$   in time $O(d(v)\cdot t\cdot\lambda)$.

%
%

In the following we   assume that an arbitrary order has been fixed on the $d=d(v)-1$ children of any node $v$, that is,  we
denote them as $v_1, v_2, \dots, v_d,$ according to the fixed order. Also, we define $F(v, i)$ to be the forest consisting of the subtrees rooted at the
first $i$ children of $v.$ 
 We will also use $F(v,i)$ to denote the set of nodes it includes.

\begin{definition} \label{def:Amax}
Let $v$ be a node with $d$ children and let $\ell =0, 1, \dots, \lambda.$
For $i=0,\ldots, d$, $j=0,1,\dots, t(v),$ we define
$A_{v, \ell}[i,j]$ (resp.\  $A_{v, \ell}[\{i\},j]$)
 to be the minimum cost for influencing all nodes in $F(v,i)$, 
(resp.\ $T_{v_i}$)
within $\lambda$ rounds, assuming that:

i) if $\ell\neq \lambda$, at time $\ell+1$ the threshold of $v_{k}$ is   $t'(v_{k}),$ for each $k = 1, \dots, i$;

ii) if $\ell\neq 0$, at least $j$ nodes in  $\{v_1,v_2,\ldots,v_i\}$  (resp. $\{v_i\}$) are influenced by round $\ell -1$, that is
\\
 \centerline{$|\{v_1,v_2,\ldots,v_i\}   \cap \Infl[\pi_{v,\ell,i,j},\ell-1]| \geq j$,} 
\hphantom{aaaaa} where $\pi_{v,\ell,i,j}: F(v, i) \rightarrow \N_0$ denotes the incentive function attaining  $A_{v, \ell}[i,j].$
\\
We also define $A_{v,\ell}[i,j]=  \infty$ 
 when the above constraints are not satisfiable.
\end{definition}
{By decomposing a solution according to how many nodes in $C(v)$ are
influenced prior to the root $v$ being influenced and denoting this number as $j$, the
remaining cost to influence the root $v$ is $t - j$ Hence, we can easily write $\MIS[v,\ell,t]$ in terms of $A_{v,\ell}[d,j]$ as follows.}

\begin{lemma} \label{prop:Amax}
For each node $v$ with $d$ children, each $\ell=0,\ldots, \lambda$   and each
$t\in \{t(v),t'(v)\}$
\begin{equation}\label{eq-Amax}
\MIS[v,\ell,t] = \begin{cases} 
t+A_{v,0}[d,0] & \text{ if }\ell=0\\
 \min_{0\leq j \leq t}  \Big\{t-j+ A_{v,\ell}[d,j]\Big\} & \text{ otherwise.}
\end{cases}
\end{equation}
\end{lemma}

\begin{proof}{}
The statement directly follows from Definitions \ref{defi:MIS} and \ref{def:Amax}. In fact when $\ell>0$ we have
\begin{eqnarray*}
\MIS[v,\ell,t]  &=& \min_{\substack{p: T_v  \to\N_0, \ \Infl_{T_v}[p,\lambda]= T_v \\ |C(v)   \cap \Infl_{F(v,d)}[p,\ell-1]| \geq t-p(v) }  } \left \{ \sum_{v\in T_v} p(v) \right\}\\
 &=& \min_{0 \leq j \leq t} \left\{  \min_{
\substack{p: T_v   \to\N_0, \ \Infl_{T_v}[p,\lambda]= T_v \\ p(v)=t-j, \ |C(v)   \cap \Infl_{F(v,d)}[p,\ell-1]| \geq j  }  
 } \left \{ \sum_{v\in T_v} p(v) \right\}\right\}\\
 &=& \min_{0 \leq j \leq t} \left\{ t-j+  \min_{\substack{p: F(v,d)  \to\N_0,  \Infl_{F(v,d)}[p,\lambda]= F(v,d) \\ |C(v)   \cap \Infl_{F(v,d)}[p,\ell-1]| \geq j }  } \left \{ \sum_{v\in  F(v,d) } p(v) \right\}\right\}\\
 &=& \min_{0 \leq j \leq t} \left\{ t-j+ A_{v,\ell}[d,j]  \right\}.
\end{eqnarray*}
where the last equality is due to the fact that the incentive on $v$ (i.e., $p(v)=t-j$) and the influence from the children (i.e., $|C(v)   \cap \Infl_{F(v,d)}[p,\ell-1]| \geq j$)  are enough to influence $v$ at round $\ell$ and consequently at round $\ell+1$ the threshold of $v_{k}$ becomes   $t'(v_{k}),$ for each $k = 1, \dots, d$.
Similarly for $\ell=0$ we have 
\begin{eqnarray*}
\MIS[v,0,t]  &=& \min_{\substack{p: T_v  \to\N_0, p(v)=t\\ \Infl_{T_v}[p,\lambda]= T_v  }  } \left \{ \sum_{v\in T_v} p(v) \right\}\\ &=&t+\min_{\substack{p:V(T_v) \to\N_0\\ \Infl_{F(v,d)}[p,\lambda]= F(v,d)  }  } \left \{ \sum_{v\in F(v,d)} p(v) \right\}\\
&=& t+A_{v,\ell}[d,0].
\end{eqnarray*}\qed
\end{proof}

\begin{lemma}\label{lemma2}
For each node $v$,  each $t{\in}\{t(v),t'(v)\}$, and each $\ell{=}1,\ldots, \lambda$, it is possible to
 compute $A_{v,\ell}[d,t]$,  as well as  $A_{v,0}[d,0]$,
 recursively in time $O(\lambda d t)$  where $d$ is the number of children of $v$.
\end{lemma}
 %

\begin{proof}{}
For $\ell=0$ the value $A_{v,0}[d,0]$ can be computed considering a complete independence among the influence processes in the different subtrees of $F(v,d)$. Moreover, for each children $v_i$ of $v$ we can consider a reduced threshold $t'(v_i)$ from round $1$. Hence we have
\begin{equation} \label{eq_7_1}
A_{v,0}[d,0]=\sum_{v_i \in C(v)} \min \left\{ \MIS[v_i,0,t(v_i)], \min_{1 \leq \ell' \leq \lambda} \{\MIS[v_i,\ell',t'(v_i)]\} \right\}
\end{equation}
For $1\leq \ell \leq \lambda$, we can  compute $A_{v,\ell}[d, t]$ by recursively computing the values  $A_{v,\ell}[i,j]$ for each $i=0,1,\ldots,d$ and  $j= 0,1,\ldots,t,$ as follows.
First we show how to compute the values $A_{v,\ell}[\{i\}, j]$.
\begin{equation}\label{Amaxsingle}
A_{v,\ell}[\{i\}, j] = \begin{cases}  \infty    & \text{ if } j>1 \\ 
\displaystyle{\min_{0 \leq \ell' \leq \ell-1} \{\MIS[v_i,\ell',t(v_i)]  \}}  & \text{ if } j=1 \\	
\displaystyle{\min  \Big\{\min_{0 \leq \ell' \leq \ell} \MIS[v_i,\ell',t(v_i)], \min_{\ell{+}1 \leq \ell' \leq \lambda} \MIS[v_i,\ell',t'(v_i)]\Big \} } & \text{ if  } j=0 
\end{cases}
\end{equation}
Indeed, for $j>1$ a single node $v_i$ can not satisfy the condition ii) of definition \ref{def:Amax}.
For $j=1$, the condition ii) of definition \ref{def:Amax} forces the influence of $v_i$, before the influence of $v$, with the original threshold  $t(v_i)$  (because in this case $v$ can not influence $v_i$).
Finally, for $j=0$ there are two possible choices. The node $v_i$ can be influenced at any round with the original threshold $t(v_i)$ or after round $\ell$ exploiting the influence of $v$ (that is using a reduced threshold $t'(v_i)$).
Recalling that $\MIS[v,\ell,t]$ is monotonically non-decreasing in $t$ we obtain the equation in (\ref{Amaxsingle}).

The following recursive equation enables to compute the values $A_{v,\ell}[d, t]$ in time $O(d t\lambda)$.
For $i=1, $ we set  $ A_{v,\ell}[1,j]=A_{v,\ell}[\{1\},j]$;
\begin{equation}\label{Amax}
A_{v,\ell}[i, j] = \begin{cases} 0    & \text{ if } i=j=0 \\ 
\infty & \text{ if } i<j\\	
\displaystyle{ \min \bigg\{ A_{v,\ell}[i{-}1,j{-}1] + \min_{0 \leq \ell' \leq \ell-1} \{\MIS[v_i,\ell',t(v_i)]\}},& \text{ otherwise. }\\
\displaystyle{\ \ A_{v,\ell}[i{-}1,j] {+} \min  \Big\{\min_{0 \leq \ell' \leq \ell} \MIS[v_i,\ell',t(v_i)], \min_{\ell{+}1 \leq \ell' \leq \lambda} \MIS[v_i,\ell',t'(v_i)]\Big \}\bigg\}}  \end{cases}
\end{equation}

Let $i \in \{2,\dots, d\}$. In order to compute $A_{v, \ell}[i, j],$ 
there are two possibilities to consider:

\begin{description}
	\item[{\textsc{I}})]  $v_i$ has to contribute to condition ii) of definition \ref{def:Amax}. Hence, $v_i$ has to be influenced before round $\ell$ and  cannot use the reduced threshold.
		\item[{\textsc{II}})] $v_i$ does not contribute to condition ii) of definition \ref{def:Amax} (i.e., the condition on the influence
	brought to $v$ from $v_1, \dots, v_i$ at round $\ell-1$ is already satisfied by $v_1, \dots, v_{i-1}).$
	In this case we have no constraint on  when $v_i$ is influenced, and we can use a reduced threshold from round $\ell+1$;
\end{description}
Therefore, for $i> 1$ and for each $i \leq j \leq t$ we can compute
$A_{v,\ell}[i,j]$ using the following formula:
\begin{equation} \label{eqBmax-tot}
A_{v,\ell}[i,j]=\min \Big\{ A_{v,\ell}^{\textsc{i}}[i,j], A_{v,\ell}^{\textsc{ii}}[i,j] \Big\}\,,
\end{equation}
where $A_{v,\ell}^{\textsc{i}}[i,j]$ and $A_{v,\ell}^{\textsc{ii}}[i,j]$ denote the corresponding optimal values of the two restricted subproblems described above.
It holds that 
$$A_{v,\ell}^{\textsc{i}}[i,j] =  A_{v,\ell}[i{-}1,j{-}1] + A_{v,\ell}[\{i\},1].$$
Hence, by equation (\ref{Amaxsingle}) we have 
\begin{equation} \label{eqBmax-b}
A_{v,\ell}^{\textsc{i}}[i,j] = A_{v,\ell}[i{-}1,j{-}1]  +\min_{0 \leq \ell' \leq \ell-1} \Big\{\MIS[v_i,\ell',t(v_i)]\Big\}.
\end{equation}
Analogously, it holds that 
$$A_{v,\ell}^{\textsc{ii}}[i,j] = A_{v, \ell}[i-1, j] + A_{v, \ell}[\{i\}, 0].$$
Hence, by equation (\ref{Amaxsingle}) we have 
\begin{eqnarray}\label{eqBmax-a}
A_{v,\ell}^{\textsc{ii}}[i,j] = \left \{ A_{v, \ell}[i-1, j] + \min  \Big\{\min_{0 \leq \ell' \leq \ell} \MIS[v_i,\ell',t(v_i)], \min_{\ell{+}1 \leq \ell' \leq \lambda} \MIS[v_i,\ell',t'(v_i)]\Big \} 
 \right\}.
\end{eqnarray}
%
%
%

\noindent
\textbf{Complexity.} From equations (\ref{Amaxsingle})-(\ref{eqBmax-a}), 
it follows that the computation of $A_{v,\ell}[\cdot,\cdot]$ comprises $O(dt)$ values and 
each one is computed recursively in time $O(\lambda)$. Hence we are able to compute it
in time $O(\lambda d t)$. \qed
\end{proof}

Lemmas \ref{prop:Amax} and \ref{lemma2}   imply that for each  $v \in V,$   for each $\ell=0,\ldots, \lambda$, and  $t \in \{t'(v),t(v)\}$, the value $\MIS[v,\ell,t]$ can be computed recursively in time $O(\lambda d(v) t(v))$. Hence, the value in (\ref{eq-mas})
 can be computed in time\\
$\sum_{v \in V} O(\lambda d(v)t(v)) \times O(\lambda)=O(\lambda^2\Delta)\times\sum_{v \in V} O(d(v))=O(\lambda^2 \Delta n),$\\
where $\Delta$ is the maximum node degree.
Standard backtracking techniques can be used to compute the (optimal) influence function $p^*$ that
influences all of the nodes in the same $O(\lambda^2 \Delta n)$ time.

\end{document}